\documentclass[11pt]{article}
\usepackage{epsf}
\usepackage{amsmath}
\usepackage{epsfig}
\usepackage{times}
\usepackage{amssymb}
\usepackage{amsthm}
\usepackage{setspace}
\usepackage{cite}

\usepackage{algorithmic}  % This package provides an algorithmic environment fo describing algorithms.
\usepackage{algorithm}

\usepackage{shadow}
\usepackage{fancybox}
\usepackage{fancyhdr}

\def\y{{\bf y}}

\def\x{{\bf x}}

\def\x{{\mathbf x}}

\def\x{{\bf x}}
\def\y{{\bf y}}
\def\z{{\bf z}}

\def\erf{\mbox{erf}}
\def\erfinv{\mbox{erfinv}}

\def\inter{{C_{int}}}
\def\exter{{C_{ext}}}
\def\com{{C_{com}}}

\def\psiint{{\Psi_{int}}}
\def\psiext{{\Psi_{ext}}}
\def\psicom{{\Psi_{com}}}
\def\psinet{{\Psi_{net}}}

\def\be{\begin{equation}}
\def\ee{\end{equation}}
\def\ba{\left[\begin{array}}
\def\ea{\end{array}\right]}

\def\x{{\bf x}}
\def\y{{\bf y}}
\def\z{{\bf z}}

\def\1{{\bf 1}}

\def\0{{\bf 0}}

\def\erfinv{\mbox{erfinv}}
\def\htheta{\hat{\theta}}

\newtheorem{theorem}{Theorem}

\begin{document}

\begin{singlespace}

\title {Linear under-determined systems with sparse solutions: Redirecting a challenge? %A tight variant of Gordon's escape through a mesh theorem
\footnote{ This work was supported in part by NSF grant \#CCF-1217857.}
}
\author{
\textsc{Mihailo Stojnic}
\\
\\
{School of Industrial Engineering}\\
{Purdue University, West Lafayette, IN 47907} \\
{e-mail: {\tt mstojnic@purdue.edu}} }
\date{}
\maketitle

\centerline{{\bf Abstract}} \vspace*{0.1in}

Seminal works \cite{CRT,DonohoUnsigned,DonohoPol} generated a massive interest in studying linear under-determined systems with sparse solutions. In this paper we give a short mathematical overview of what was accomplished in last $10$ years in a particular direction of such a studying. We then discuss what we consider were the main challenges in last $10$ years and give our own view as so what are the main challenges that lie ahead. Through the presentation we arrive to a point where the following natural rhetoric question arises: is it a time to redirect the main challenges? While we can not provide the answer to such a question we hope that our small discussion will stimulate further considerations in this direction.

\vspace*{0.25in} \noindent {\bf Index Terms: Linear systems of equations; sparse solutions;
$\ell_1$-optimization}.

\end{singlespace}

%%%%%%%%%%%%%%%%%%%%%%%%%%%%%%%%%%%%%%%%%%%%%%%%%%%%%%%%%%%%%%%%%
\section{Introduction}
\label{sec:back}
%%%%%%%%%%%%%%%%%%%%%%%%%%%%%%%%%%%%%%%%%%%%%%%%%%%%%%%%%%%%%%%%%

In this paper we will be interested in studying under-determined systems of linear equations with sparse solutions. We start by looking at the mathematical formulation of such a problem that attracted enormous attention in recent years. The problem is essentially the following. Let $\tilde{\x}$ be an $n$-dimensional vector from $R^n$. Moreover let $\tilde{\x}$ be $k$-sparse (under $k$-sparse we assume vectors that have at most $k$ components that are not equal to zero; clearly $k\leq n$). Let $A$ be an $m \times n$ matrix from $R^{m\times n}$. We will call $A$ the system or the measurement matrix and throughout the rest of the paper assume that $A$ is a full rank matrix (on occasions when $A$ happens to be random we will assume that $A$ is of full rank with overwhelming probability, where under overwhelming probability we consider a probability that is not more than a number exponentially decaying in $n$ away from $1$).
Now, the question of interest is: given $A$ and $A\tilde{\x}$ can one find $\x$ such that
\begin{equation}
A\x=A\tilde{\x}.\label{eq:system}
\end{equation}
 Fairly often one refers to $A\tilde{\x}$ in (\ref{eq:system}) as a known vector $\y$ from $R^m$. In other words, if one rewrite the system given in \ref{eq:system} in a more natural way
\begin{equation}
A\x=\y.\label{eq:systemy}
\end{equation}
$y$ is essentially implied to be constructed as the product of matrix $A$ and a $k$-sparse vector $\tilde{\x}$. We will often in the rest of the paper use the expression ``solve the system give in (\ref{eq:system})". By that we will mean that if $\hat{\x}$ is the solution we found by solving (\ref{eq:systemy}) using any available methodology then $\hat{\x}=\tilde{\x}$.

The problem stated in (\ref{eq:system}) (or in (\ref{eq:systemy}) is very simple. In fact as we mentioned above it is nothing but a system of linear equations which are additionally assumed to have sparse solutions. As is usually the case with linear systems, a critical piece of information that enables one to solve the problem is the relation between $k$, $m$, and $n$. Clearly, if $m\geq n$ the system is either over-determined or just determined and solving (\ref{eq:system}) could a bit easier. On the other hand if $m<n$ the system is under-determined and in general may not have a unique solution. However, if $k<m$ one may still be able to figure out what $\tilde{\x}$ is. Clearly, if one knows a priori the value of $k$ one can then search over all subsystems obtained by extracting $k$ columns from matrix $A$. Of course, such an approach would probably solve the problem but it is very complex when $n$ and $k$ are large. To be a bit more specific (as well as to make all our points in the paper clearer) we will in the rest of the paper assume the so-called linear regime, i.e. the regime where all $k$, $m$, and $n$ are large but proportional to each other. To be even more specific we will assume that constant of proportionality are $\beta$ and $\alpha$, i.e. we will assume that $\beta=\frac{k}{n}$ and $\alpha=\frac{m}{n}$. Under such an assumption the complexity of the above mentioned strategy of extracting all subsystems with $k$ columns would be exponential. Instead of such a simple strategy one can employ a host of more sophisticated approaches. Since this paper has its own objective we will not present all known approaches. Instead will try to shorten our presentation in that regard and focus only on what we consider are the most popular/well known/successful ones. Moreover, since studying any of such approaches is nowadays pretty much a theory on its own we will most often just give the core information and leave more sophisticated discussion for overview type of the papers.

We start by emphasizing that one can generally distinguish two classes of possible algorithms that can be
developed for solving (\ref{eq:system}). The first class of algorithms assumes freedom in designing matrix $A$. Such a class is already a bit different from what can be employed for solving (\ref{eq:system}) or (\ref{eq:systemy}). Namely, as we mentioned right before (\ref{eq:system}) and (\ref{eq:systemy}), our original setup assumes that we are given a matrix $A$. Still to maintain a completeness of the exposition we briefly mention this line of work and do so especially since what can be achieved within such an approach seems to be substantially better than what can be achieved within our setup.

So, if one has the freedom to design matrix $A$ then the results from \cite{FHicassp,Tarokh,MaVe05} demonstrated that the techniques from
coding theory (based on the coding/decoding of Reed-Solomon codes)
can be employed to determine \emph{any} $k$-sparse $\x$ in
(\ref{eq:system}) for any $0<\alpha\leq 1$ and any
$\beta\leq\frac{\alpha}{2}$ in polynomial time (it is relatively easy to show that under the unique recoverability assumption
$\beta$ can not be greater than $\frac{\alpha}{2}$). Therefore, as long as one is concerned with the unique recovery of
$k$-sparse $\x$ in (\ref{eq:system}) in polynomial time the results from \cite{FHicassp,Tarokh,MaVe05} are
optimal. The complexity of algorithms from
\cite{FHicassp,Tarokh,MaVe05} is roughly $O(n^3)$. In a similar fashion one can, instead of using coding/decoding techniques associated with Reed/Solomon codes,
design the matrix and the corresponding recovery algorithm based on the techniques related to the coding/decoding of
Expander codes (see e.g.
\cite{XHexpander,JXHC08,InRu08} and references therein). In that case recovering $\x$ in
(\ref{eq:system}) is significantly faster for large dimensions $n$. Namely, the complexity of the techniques from e.g. \cite{XHexpander,JXHC08,InRu08}
(or their slight modifications) is usually
$O(n)$ which is clearly for large $n$ significantly smaller than $O(n^3)$. However,
the techniques based on coding/decoding of Expander codes usually do not allow for $\beta$ to be as large as
$\frac{\alpha}{2}$.

The main interest of this paper however will be the algorithms from the second class. Within the second class are the algorithms that should be designed without having the choice of $A$ (instead, as mentioned right before (\ref{eq:system}) and (\ref{eq:systemy}) matrix $A$ is rather given to us). Designing the algorithms from the second class is substantially harder compared to the design of the algorithms from the first class. The main reason for hardness is that when there is no choice in $A$ the recovery
problem (\ref{eq:system}) becomes NP-hard. The following three algorithms (and their different
variations) we currently view as solid heuristics for solving (\ref{eq:system}):
\begin{enumerate}
\item \underline{\emph{Orthogonal matching pursuit - OMP}}
\item \underline{\emph{Basis pursuit -
$\ell_1$-optimization.}}
\item \underline{\emph{Approximate message passing - AMP}}
\end{enumerate}
We do however mention that in the third technique, which is based on belief propagation type of algorithms, is emerging as a strong alternative in recent years. While it does not have as strong historical a background as the other two (at least when it comes to solving (\ref{eq:system})) its great performance features as well as its a relatively easy implementation make it particularly attractive.
Under certain probabilistic assumptions on the elements of $A$ it can be shown (see e.g. \cite{JATGomp,JAT,NeVe07})
that if $m=O(k\log(n))$
OMP (or slightly modified OMP) can recover $\x$ in (\ref{eq:system})
with complexity of recovery $O(n^2)$. On the other hand a stage-wise
OMP from \cite{DTDSomp} recovers $\x$ in (\ref{eq:system}) with
complexity of recovery $O(n \log n)$. Somewhere in between OMP and BP are recent improvements CoSAMP (see e.g. \cite{NT08}) and Subspace pursuit (see e.g. \cite{DaiMil08}), which guarantee (assuming the linear regime) that the $k$-sparse $\x$ in (\ref{eq:system}) can be recovered in polynomial time with $m=O(k)$ equations.

%%%%%%%%%%%%%%%%%%%%%%%%%%%%%%%%%%%%%%%%%%%%%%%%%%%%%%%%%%%%%%%%%
\section{$\ell_1$-optimization}
\label{sec:l1}
%%%%%%%%%%%%%%%%%%%%%%%%%%%%%%%%%%%%%%%%%%%%%%%%%%%%%%%%%%%%%%%%%

We will now further narrow down our interest to only the performance of $\ell_1$-optimization. (Variations of the standard $\ell_1$-optimization from e.g.
\cite{CWBreweighted,SChretien08,SaZh08}) as well as those from \cite{SCY08,FL08,GN03,GN04,GN07,DG08} related to $\ell_q$-optimization, $0<q<1$
are possible as well.) Basic $\ell_1$-optimization algorithm offers an $\x$ in
(\ref{eq:system}) as the solution of the following $\ell_1$-norm minimization problem
\begin{eqnarray}
\mbox{min} & & \|\x\|_{1}\nonumber \\
\mbox{subject to} & & A\x=\y. \label{eq:l1}
\end{eqnarray}
Due to its popularity the literature on the use of the above algorithm is rapidly growing. We below restrict our attention to two, in our mind, the most influential works that relate to (\ref{eq:l1}).

The first one is \cite{CRT} where the authors were able to show that if
$\alpha$ and $n$ are given, $A$ is given and satisfies the restricted isometry property (RIP) (more on this property the interested reader can find in e.g. \cite{Crip,CRT,Bar,Ver,ALPTJ09}), then
any unknown vector $\x$ with no more than $k=\beta n$ (where $\beta$
is a constant dependent on $\alpha$ and explicitly
calculated in \cite{CRT}) non-zero elements can be recovered by
solving (\ref{eq:l1}). As expected, this assumes that $\y$ was in
fact generated by that $\x$ and given to us. The case when the
available $\y$'s are noisy versions of real $\y$'s is also of
interest \cite{CT,CRT,HN,W}. Although that case is not of primary interest in the present paper
it is worth mentioning that the recent
popularity of $\ell_1$-optimization in the field of compressed sensing (where problem (\ref{eq:system}) is one of key importance) is
significantly due to its robustness with respect to noisy
$\y$'s. (Of course, the main reason for its popularity is its
ability to solve (\ref{eq:system}) for a very wide range of matrices
$A$; more on this universality from a statistical point of view the interested reader can find in \cite{DonTan09Univ}.)

However, the RIP is only a \emph{sufficient}
condition for $\ell_1$-optimization to produce the $k$-sparse solution of
(\ref{eq:system}). Instead of characterizing $A$ through the RIP
condition, several alternative route have been introduced in recent years. Among the most successful ones are those from e.g.  \cite{DonohoUnsigned,DonohoPol,DonohoSigned,DT,StojnicCSetam09,StojnicUpper10} and we will revisit them below. However, before revisiting these approaches we should mention that it was fairly early observed that if matrix $A$ and vector $\tilde{\x}$ are deterministic (and hence can always be chosen so to make the solution of (\ref{eq:l1}) be as far away as possible from $\tilde{\x}$) then it is highly unlikely that (\ref{eq:l1}) would be of much help in providing a provably fast (say, polynomial) way for ``guaranteed" solving of (\ref{eq:system}). Having this in mind the shift to statistical $A$ and/or $\tilde{\x}$ happened fairly quickly. The idea for such a shift can be summarized in the following way: if it is not possible to recover $\tilde{\x}$ in (\ref{eq:system}) by solving (\ref{eq:l1}) for \emph{all} $A$ and $\tilde{\x}$ then maybe it can still be possible for an overwhelming majority of them. A way to characterize such an overwhelming majority is then to introduce randomness on $A$ and/or $\tilde{\x}$. For example, if $A$ is a random matrix one may be able to say that if a concrete $A$ (i.e., its elements) in (\ref{eq:system}) is drawn from a probabilistic distribution then maybe for such an $A$ the solution of (\ref{eq:l1}) is often $\tilde{\x}$. This is somewhat standard way of attacking NP-hardness (there are of course more sophisticated ways, but in this paper we will look just at this basic premise).

%%%%%%%%%%%%%%%%%%%%%%%%%%%%%%%%%%%%%%%%%%%%%%%%%%%%%%%%%%%%%%%%%
\subsection{Geometric approach to $\ell_1$-optimization}
\label{sec:geoml1}
%%%%%%%%%%%%%%%%%%%%%%%%%%%%%%%%%%%%%%%%%%%%%%%%%%%%%%%%%%%%%%%%%

In \cite{DonohoUnsigned,DonohoPol} Donoho revisited the $\ell_1$-optimization technique from (\ref{eq:l1}) and looked at its geometric properties/potential. Namely,
in \cite{DonohoUnsigned,DonohoPol} Donoho considered the polytope obtained by
projecting the regular $n$-dimensional cross-polytope $C_p^n$ by $A$. He then established that
the solution of (\ref{eq:l1}) will be the $k$-sparse solution of
(\ref{eq:system}) (i.e., it will be $\tilde{\x}$) if and only if
$AC_p^n$ is centrally $k$-neighborly
(for the definitions of neighborliness, details of Donoho's approach, and related results the interested reader can consult now already classic references \cite{DonohoUnsigned,DonohoPol,DonohoSigned,DT}). In a nutshell, relying on a long line of geometric results
from \cite{PMM,AS,BorockyHenk,Ruben,VS}, in
\cite{DonohoPol} Donoho showed that if $A$ is a random $m\times n$
ortho-projector matrix then with overwhelming probability $AC_p^n$ is centrally $k$-neighborly (as mentioned earlier, under overwhelming probability we in this paper assume
a probability that is no more than a number exponentially decaying in $n$ away from $1$). Miraculously, \cite{DonohoPol,DonohoUnsigned} provided a precise characterization of $m$ and $k$ (in a large dimensional context) for which this happens.

Before, presenting the details of Donoho's findings we should make a few clarifications. Namely, it should be noted that one usually considers success of
(\ref{eq:l1}) in recovering \emph{any} given $k$-sparse $\tilde{\x}$ in (\ref{eq:system}). It is also of interest to consider success of
(\ref{eq:l1}) in recovering
\emph{almost any} given $\x$ in (\ref{eq:system}). We below make a distinction between these
cases and recall on some of the definitions from
\cite{DonohoPol,DT,DTciss,DTjams2010,StojnicCSetam09,StojnicICASSP09}.

Clearly, for any given constant $\alpha\leq 1$ there is a maximum
allowable value of $\beta$ such that for \emph{any} given $k$-sparse $\tilde{\x}$ in (\ref{eq:system}) the solution of (\ref{eq:l1})
is with overwhelming probability exactly that given $k$-sparse $\tilde{\x}$. We will refer to this maximum allowable value of
$\beta$ as the \emph{strong threshold} (see
\cite{DonohoPol}). Similarly, for any given constant
$\alpha\leq 1$ and \emph{any} given $\x$ with a given fixed location of non-zero components and a given fixed combination of its elements signs
there will be a maximum allowable value of $\beta$ such that
(\ref{eq:l1}) finds that given $\x$ in (\ref{eq:system}) with overwhelming
probability. We will refer to this maximum allowable value of
$\beta$ as the \emph{weak threshold} and will denote it by $\beta_{w}$ (see, e.g. \cite{StojnicICASSP09,StojnicCSetam09}). What we present below are essentially Donoho's findings that relate to $\beta_w$.

Knowing all of this, we can then state what was established in \cite{DonohoPol}. If $A$ is a random ortho-projector $AC_p^n$ will be centrally $k$-neighborly with overwhelming probability if
\begin{equation}
n^{-1}\log(\com\inter(T^{k},T^{m})\exter(F^{m},C_p^n))<0 \label{eq:conddonpol}
\end{equation}
where $\com=2^{m-k}\binom{n-k-1}{m-k}$, $\inter(T^{k},T^{m})$ is the internal angle at face $T^{k}$ of $T^{m}$, $\exter(F^{m},C_p^n)$ is the external angle of $C_p^n$ at any $m$-dimensional face $F^m$, and $T^{k}$ and $T^{m}$ are the standard $k$ and $m$ dimensional simplices, respectively (more on the definitions and meaning of the internal and external angles can be found in e.g. \cite{Grunbaum03}). Donoho then proceeded by establishing that (\ref{eq:conddonpol}) is equivalent to the following inequality related to the sum/difference of the exponents of $\com,\inter$, and $\exter$:
\begin{equation}
\psinet=\psicom-\psiint-\psiext<0 \label{eq:conddonpolexp}
\end{equation}
where
\begin{eqnarray}
\psicom(\beta,\alpha) & = & n^{-1}\log(\com)=(\alpha-\beta) \log(2)+(1-\beta)H(\frac{\alpha-\beta}{1-\beta})\nonumber \\
\psiint(\beta,\alpha) & = & n^{-1}\log(\inter(T^{k},T^{m})) \nonumber\\
\psiext(\beta,\alpha) & = & n^{-1}\log(\exter(F^{m},C_p^n)) \label{eq:conddonpolexp1}
\end{eqnarray}
and $H(p)=-p\log(p)-(1-p)\log(1-p)$ is the standard entropy function and $\log\binom{n}{pn}=e^{nH(p)}$ is the standard approximation of the binomial factor by the entropy function in the limit of $n\rightarrow\infty$. Moreover, Donoho also provided a way to characterize $\psiint(\beta,\alpha),\psiext(\beta,\alpha))$. Let $\gamma=\frac{\beta}{\alpha}$ and for $s\geq 0$
\begin{eqnarray}
\Phi(s) & = & \frac{1}{\sqrt{2\pi}}\int_{s}^{\infty}e^{-\frac{x^2}{2}}dx\nonumber \\
\phi(s) & = & \frac{1}{\sqrt{2\pi}}e^{-\frac{s^2}{2}}.\label{eq:phis}
\end{eqnarray}
Then one has
\begin{equation}
\psiint(\beta,\alpha)=(\alpha-\beta)\xi_{\gamma}(y_{\gamma})+(\alpha-\beta)\log(2)\label{eq:intang1}
\end{equation}
where
\begin{eqnarray}
y_{\gamma} & = & \frac{\gamma}{1-\gamma}s_{\gamma}\nonumber \\
\xi_{\gamma}(y_{\gamma}) & =  & -\frac{1}{2}y_{\gamma}^2\frac{1-\gamma}{\gamma}-\frac{1}{2}\log(\frac{2}{\pi})+\log(\frac{y_{\gamma}}{\gamma}) \label{eq:intang2}
\end{eqnarray}
and $s_{\gamma}\geq 0$ is the solution of
\begin{equation}
\Phi(s)=(1-\gamma)\frac{\phi(s)}{s}.\label{eq:intang3}
\end{equation}
On the other hand the expression for $\psiext(\beta,\alpha)$ is a bit simpler
\begin{equation}
\psiext(\beta,\alpha)=\min_{y\geq 0} (\alpha y^2 -(1-\alpha)\log(\erf(y))).\label{eq:extang1}
\end{equation}
Using (\ref{eq:conddonpolexp1}), (\ref{eq:phis}), (\ref{eq:intang1}), (\ref{eq:intang2}), (\ref{eq:intang3}), (\ref{eq:extang1}) one then for a fixed $\alpha$ finds the largest $\beta$ so that the left-hand side of (\ref{eq:conddonpolexp}) is basically zero. Such a $\beta$ is what we termed above as $\beta_w$. While the above characterization of optimal $\beta_w$ (as a function of $\alpha$) is not super simple it is truly fascinating that it actually ends up being exact.

We summarize the above results in the following theorem.

\begin{theorem}(Exact $\ell_1$ $(\beta_w,\alpha_w)$ threshold --- geometric approach of \cite{DonohoUnsigned,DonohoPol})
Let $A$ in (\ref{eq:system}) be an $m\times n$ ortho-projector (or an $m\times n$ matrix
with the null-space uniformly distributed in the Grassmanian).
Let $k,m,n$ be large
and let $\alpha_w=\frac{m}{n}$ and $\beta=\frac{k}{n}$ be constants
independent of $m$ and $n$. Let $\psicom(\beta,\alpha_w),\psiint(\beta,\alpha_w),\psiext(\beta,\alpha_w)$ be evaluated for a pair $(\beta,\alpha_w)$ through the expressions given in (\ref{eq:conddonpolexp1}), (\ref{eq:intang1}), and (\ref{eq:extang1}). Let then $\beta_w$ be the maximal $\beta$ for which (\ref{eq:conddonpolexp}) holds. Then:

1) With overwhelming
probability polytope $AC_p^n$ will be centrally $\beta n$-neighborly for any $\beta<\beta_w$.

2) With overwhelming
probability polytope $AC_p^n$ will not be centrally $\beta n$-neighborly for any $\beta>\beta_w$.
Moreover, let
$\tilde{\x}$ in (\ref{eq:system}) be $k$-sparse. Then:

1) If $\beta<\beta_w$ then with overwhelming
probability for almost any $\tilde{\x}$, the solution of (\ref{eq:l1}) is exactly that $\tilde{\x}$.

2) If $\beta>\beta_w$ then with overwhelming
probability for almost any $\tilde{\x}$, the solution of (\ref{eq:l1}) is not that $\tilde{\x}$.
\label{thm:thmweakthrdon}
\end{theorem}
\begin{proof}
Follows from considerations presented in \cite{DonohoUnsigned,DonohoPol}.
\end{proof}

In the following subsection we present a related collection of results that were obtained in a series of our own work \cite{StojnicCSetam09,StojnicUpper10,StojnicICASSP09} attacking performance analysis of (\ref{eq:l1}) through a probabilistic approach.

%%%%%%%%%%%%%%%%%%%%%%%%%%%%%%%%%%%%%%%%%%%%%%%%%%%%%%%%%%%%%%%%%
\subsection{Purely probabilistic approach to $\ell_1$-optimization}
\label{sec:probl1}
%%%%%%%%%%%%%%%%%%%%%%%%%%%%%%%%%%%%%%%%%%%%%%%%%%%%%%%%%%%%%%%%%

In our own work \cite{StojnicCSetam09} we introduced a novel probabilistic framework for performance characterization of (\ref{eq:l1}) (the framework seems rather powerful; in fact, we found hardly any sparse type of problem that the framework was not able to handle with almost impeccable precision). Using that framework we obtained lower bounds on $\beta_w$. These lower bounds were in an excellent numerical agreement with the values obtained for $\beta_w$ in \cite{DonohoPol}. We were therefore tempted to believe that our lower bounds from \cite{StojnicCSetam09} are tight. In a follow up paper \cite{StojnicUpper10} we then presented a mechanism that can be used to obtain matching upper-bounds, therefore establishing formally results from \cite{StojnicCSetam09} as an alternative ultimate performance characterization of (\ref{eq:l1}). Alternatively, in \cite{StojnicEquiv10}, we provided a rigorous analytical matching of $\beta_w$ threshold characterizations from \cite{StojnicCSetam09} and those given in \cite{DonohoUnsigned}. The following theorem summarizes the results we obtained in \cite{StojnicICASSP09,StojnicCSetam09,StojnicUpper10,StojnicEquiv10}.

\begin{theorem}(Exact $\ell_1$ $(\beta_w,\alpha_w)$ threshold -- probabilistic approach of \cite{StojnicCSetam09,StojnicUpper10})
Let $A$ be an $m\times n$ matrix in (\ref{eq:system})
with i.i.d. standard normal components. Let
$\tilde{\x}$ in (\ref{eq:system}) be $k$-sparse. Further, let the location and signs of nonzero elements of $\tilde{\x}$ be arbitrarily chosen but fixed.
Let $k,m,n$ be large
and let $\alpha=\frac{m}{n}$ and $\beta_w=\frac{k}{n}$ be constants
independent of $m$ and $n$. Let $\erfinv$ be the inverse of the standard error function associated with zero-mean unit variance Gaussian random variable.  Further,
let all $\epsilon$'s below be arbitrarily small constants.
\begin{enumerate}
\item Let $\htheta_w$, ($\beta_w\leq \htheta_w\leq 1$) be the solution of
\begin{equation}
(1-\epsilon_{1}^{(c)})(1-\beta_w)\frac{\sqrt{\frac{2}{\pi}}e^{-(\erfinv(\frac{1-\theta_w}{1-\beta_w}))^2}}{\theta_w}-\sqrt{2}\erfinv ((1+\epsilon_{1}^{(c)})\frac{1-\theta_w}{1-\beta_w})=0.\label{eq:thmweaktheta}
\end{equation}
If $\alpha$ and $\beta_w$ further satisfy
\begin{equation}
\alpha>\frac{1-\beta_w}{\sqrt{2\pi}}\left (\sqrt{2\pi}+2\frac{\sqrt{2(\erfinv(\frac{1-\htheta_w}{1-\beta_w}))^2}}{e^{(\erfinv(\frac{1-\htheta_w}{1-\beta_w}))^2}}-\sqrt{2\pi}
\frac{1-\htheta_w}{1-\beta_w}\right )+\beta_w
-\frac{\left ((1-\beta_w)\sqrt{\frac{2}{\pi}}e^{-(\erfinv(\frac{1-\hat{\theta}_w}{1-\beta_w}))^2}\right )^2}{\hat{\theta}_w}\label{eq:thmweakalpha}
\end{equation}
then with overwhelming probability the solution of (\ref{eq:l1}) is the $k$-sparse $\tilde{\x}$ from (\ref{eq:system}).
\item Let $\htheta_w$, ($\beta_w\leq \htheta_w\leq 1$) be the solution of
\begin{equation}
(1+\epsilon_{2}^{(c)})(1-\beta_w)\frac{\sqrt{\frac{2}{\pi}}e^{-(\erfinv(\frac{1-\theta_w}{1-\beta_w}))^2}}{\theta_w}-\sqrt{2}\erfinv ((1-\epsilon_{2}^{(c)})\frac{1-\theta_w}{1-\beta_w})=0.\label{eq:thmweaktheta1}
\end{equation}
If on the other hand $\alpha$ and $\beta_w$ satisfy
\begin{multline}
\hspace{-.5in}\alpha<\frac{1}{(1+\epsilon_{1}^{(m)})^2}\left ((1-\epsilon_{1}^{(g)})(\htheta_w+\frac{2(1-\beta_w)}{\sqrt{2\pi}} \frac{\sqrt{2(\erfinv(\frac{1-\htheta_w}{1-\beta_w}))^2}}{e^{(\erfinv(\frac{1-\htheta_w}{1-\beta_w}))^2}})
-\frac{\left ((1-\beta_w)\sqrt{\frac{2}{\pi}}e^{-(\erfinv(\frac{1-\hat{\theta}_w}{1-\beta_w}))^2}\right )^2}{\hat{\theta}_w(1+\epsilon_{3}^{(g)})^{-2}}\right )\label{eq:thmweakalpha}
\end{multline}
then with overwhelming probability there will be a $k$-sparse $\tilde{\x}$ (from a set of $\tilde{\x}$'s with fixed locations and signs of nonzero components) that satisfies (\ref{eq:system}) and is \textbf{not} the solution of (\ref{eq:l1}).
\end{enumerate}
\label{thm:thmweakthrstoj}
\end{theorem}
\begin{proof}
The first part was established in \cite{StojnicCSetam09} and the second one was established in \cite{StojnicUpper10}. An alternative way of establishing the same set of results was also presented in \cite{StojnicEquiv10}.
\end{proof}

We below provide a more informal interpretation of what was established by the above theorem. Assume the setup of the above theorem. Let $\alpha_w$ and $\beta_w$ satisfy the following:

\vspace{.7in}
\noindent \underline{\underline{\textbf{Fundamental characterization of the $\ell_1$ performance:}}}

\begin{center}
\shadowbox{$
%\begin{equation}
(1-\beta_w)\frac{\sqrt{\frac{2}{\pi}}e^{-(\erfinv(\frac{1-\alpha_w}{1-\beta_w}))^2}}{\alpha_w}-\sqrt{2}\erfinv (\frac{1-\alpha_w}{1-\beta_w})=0.
%\end{equation}
$}
-\vspace{-.5in}\begin{equation}
\label{eq:thmweaktheta2}
\end{equation}
\end{center}

Then:

1) If $\alpha>\alpha_w$ then with overwhelming probability the solution of (\ref{eq:l1}) is the $k$-sparse $\tilde{\x}$ from (\ref{eq:system}).

2) If $\alpha<\alpha_w$ then with overwhelming probability there will be a $k$-sparse $\tilde{\x}$ (from a set of $\tilde{\x}$'s with fixed locations and signs of nonzero components) that satisfies (\ref{eq:system}) and is \textbf{not} the solution of (\ref{eq:l1}).

As mentioned above, in \cite{StojnicEquiv10} we established that the characterization given in Theorems \ref{thm:thmweakthrdon} and \ref{thm:thmweakthrstoj} are analytically equivalent which essentially makes
(\ref{eq:thmweaktheta2}) the ultimate performance characterization of $\ell_1$-optimization when it comes to its use in finding the sparse solutions of random under-determined linear systems.

%%%%%%%%%%%%%%%%%%%%%%%%%%%%%%%%%%%%%%%%%%%%%%%%%%%%%%%%%%%%%%%%%
\section{Approximate message passing - AMP}
\label{sec:amp}
%%%%%%%%%%%%%%%%%%%%%%%%%%%%%%%%%%%%%%%%%%%%%%%%%%%%%%%%%%%%%%%%%

In this section we briefly revisit a novel approach for solving (\ref{eq:system}). The approach was introduced in \cite{DonMalMon09}. It is essentially an iterative algorithm:
\begin{eqnarray}
\x^{(t+1)} & = & \eta_t(A^T\z^{(t)}+\x^{(t)})\nonumber \\
\z^{(t)} & = & \y-A\x^{(t)}+\frac{1}{\alpha}\z^{(t-1)}\mbox{Avg}(\eta_{t}^{'}(A^T\z^{(t-1)}+\x^{(t-1)})).\label{eq:amp}
\end{eqnarray}
$\eta_t$ is a scalar function which operates component-wise on vectors and $\eta_t'$ is the first derivative of $\eta_t$ with respect to its scalar argument. $\mbox{Avg}$ is a function that computes the average value of the components of its vector argument. The algorithm is iterative and a stopping criterion should be specified as well. There are many ways how this can be done; for example one can stop the algorithm when a norm of the difference between two successive $\x$'s is what one deems small when compared to their own norms. The more important question is why this algorithm would have a good performance. In the absence of term $\frac{1}{\alpha}\z^{(t-1)}\mbox{Avg}(\eta_{t}^{'}(A^T\z^{(t-1)}+\x^{(t-1)}))$ the algorithm boils down to the class of iterative thresholding algorithms considered in e.g. \cite{MalDon10}. These algorithms have a solid recovery abilities and are very fast. The algorithm (\ref{eq:amp}) is obviously also very easy to implement and has a substantially lower running complexity than BP. Using a state evolution formalism in \cite{DonMalMon09} a fairly precise performance characterization of (\ref{eq:amp}) when used for finding $\tilde{\x}$ in (\ref{eq:system}) was given. Namely, in \cite{DonMalMon09} the authors established that
\begin{equation}
\beta_w^{(amp)}=\alpha_w^{(amp)}\max_{z\geq 0}\left ( \frac{1-2/\alpha_w^{(amp)}((1+z^2)\Phi(z)-z\phi(z))}{1+z^2-2((1+z^2)\Phi(z)-z\phi(z))}\right ),\label{eq:betaseamp}
\end{equation}
with $\beta_w^{(amp)}$ and $\alpha_w^{(amp)}$ having meanings similar to those of $\beta_w$ and $\alpha_w)$ from the previous section. Moreover, in \cite{BayMon10} the state evolution formalism was proved to hold thereby establishing findings of \cite{DonMalMon09} as rigorous.

We summarize the above results in the following theorem.
\begin{theorem}(Exact AMP $(\beta_w^{(amp)},\alpha_w^{(amp)})$ threshold  --- AMP approach of \cite{DonMalMon09,BayMon10})
Let $A$ be an $m\times n$ matrix in (\ref{eq:system})
with i.i.d. standard normal components. Let
$\tilde{\x}$ in (\ref{eq:system}) be $k$-sparse and given. Let $k,m,n$ be large
and let $\alpha_w^{(amp)}=\frac{m}{n}$ and $\beta=\frac{k}{n}$ be constants
independent of $m$ and $n$. Let $\Phi(z)$ and $\phi(z)$ be as defined in (\ref{eq:phis}). Let $\beta_w^{(amp)}$ be as defined in (\ref{eq:betaseamp}). Then there is a suitable function $\eta_t$ in (\ref{eq:amp}) (e.g. a properly tuned simple soft thresholding function would suffice) such that:

1) If $\beta<\beta_w^{(amp)}$ the solution of (\ref{eq:amp}) is the $k$-sparse $\tilde{\x}$ in (\ref{eq:system}) with overwhelming probability.

2) If $\beta>\beta_w^{(amp)}$ the solution of (\ref{eq:amp}) is not the $k$-sparse $\tilde{\x}$ in (\ref{eq:system}) with overwhelming probability.
\label{thm:thmweakthrDMM}
\end{theorem}
\begin{proof}
The algorithm as well as the general finding were established in \cite{DonMalMon09}. The mathematical correctness was established in \cite{BayMon10}.
\end{proof}
Moreover, in \cite{DonMalMon09} it was established that the characterization given in (\ref{eq:betaseamp}) actually analytically matches the characterization given in (\ref{eq:thmweaktheta2}) (and based on findings of \cite{StojnicEquiv10} automatically the one obtained by Donoho and given in Theorem \ref{thm:thmweakthrdon}). All in all, based on everything we mentioned above one is essentially left with a signle characterization that determines performance of both, the $\ell_1$-optimization algorithm from (\ref{eq:l1}) and the AMP algorithm from (\ref{eq:amp}). Below, in Figure \ref{fig:weak} we present the characterization in $(\beta,\alpha)$ plane.
\begin{figure}[htb]
%%%%%\begin{minipage}[b]{1.0\linewidth}
\centering
\centerline{\epsfig{figure=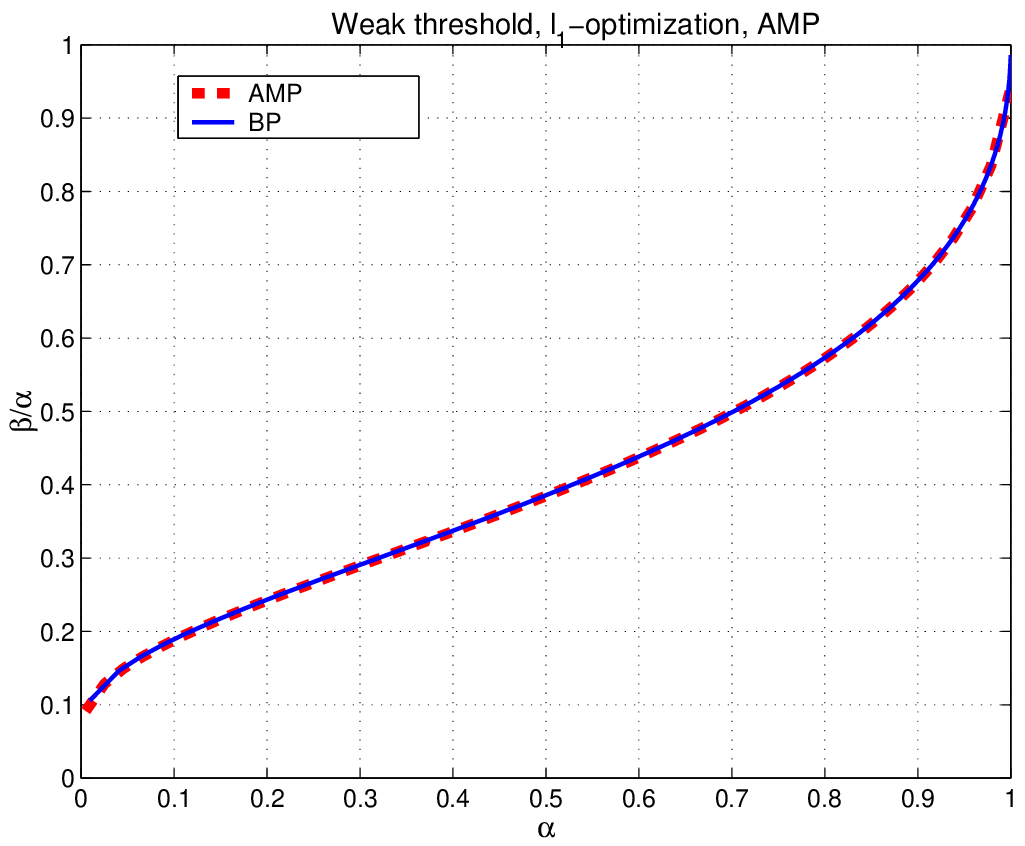,width=10.5cm,height=9cm}}
%%%%%%\end{minipage}
\caption{\emph{Weak} threshold, $\ell_1$-optimization (BP), AMP}
\label{fig:weak}
\end{figure}

%%%%%%%%%%%%%%%%%%%%%%%%%%%%%%%%%%%%%%%%%%%%%%%%%%%%%%%%%%%
\section{Revisiting challenges} \label{sec:challenge}
%%%%%%%%%%%%%%%%%%%%%%%%%%%%%%%%%%%%%%%%%%%%%%%%%%%%%%%%%%%

In the previous sections we revisited algorithmic and theoretical results that we view as the most successful currently available when it comes to recovering $\tilde{\x}$ in (\ref{eq:system}). However, if one now looks at the timeline of all these results one can observe that since the original work of Donoho \cite{DonohoUnsigned,DonohoPol} appeared almost $10$ years ago not much changed in the performance characterizations. Of course not much can be changed, Donoho actually determined the performance characterization of the $\ell_1$-optimization. What we really mean when we say not much has changed is that there has not been alternative characterizations that go above the one presented in Figure \ref{fig:weak}. One can object even this statement. Namely, there are various special cases when the characterizations can be lifted (see, e.g. \cite{StojnicICASSP10knownsupp} or say various other papers that deal with reweighted $\ell_1$ type of algorithms \cite{CWBreweighted,SChretien08,SaZh08}). Still, while there are scenarios where the characterizations can be lifted, we have not seen yet what we would consider a ``universal" lift of the characterization given in Figure \ref{fig:weak}. When we say universal, we actually mean that the characterization should faithfully portray an algorithm's performance over a fairly uniform choice of $\tilde{\x}$ (or even over all $\tilde{\x}$). For example, while all reweighted versions of $\ell_1$ typically provide substantial improvement over $\ell_1$ they typically fail to do so when $\tilde{\x}$ has binary nonzero components. This of course raises a question as to what one can/should consider as a universal improvement over $\ell_1$ and a fairly uniform choice of $\tilde{\x}$. In our view, to quantify uniformity of $\tilde{\x}$ for which we expect algorithms to work the characterization formulation given in Theorem \ref{thm:thmweakthrstoj} could be somewhat useful. Namely, borrowing parts of a setup of such a formulation one can pose the following problem:

\textbf{Question 1}:
Let $A$ be an $\alpha n\times n$ matrix with i.i.d standard normal components. Let $\tilde{\x}$ be a $\beta n$ -sparse $n$-dimensional vector from $R^n$ and let the signs and locations of its non-zero components be arbitrarily chosen but fixed. Moreover, let pair $(\beta,\alpha)$ reside in the area above the curve given in Figure \ref{fig:weak}. Can one then design a polynomial algorithm that would with overwhelming probability (taken over randomness of $A$) solve (\ref{eq:system}) for all such $\tilde{\x}$?

The idea is then that if the answer to the above question is yes then we ``agree" that an improvement over $\ell_1$ has been made. Of course, it is not really clear if the above question is really the best possible to assess a potential improvement over $\ell_1$. Essentially, in our view, it is an individual assessment what establishes an improvement and what does not. For us for example, it is actually even hard to explain what we would consider an improvement. Since this is a mathematical paper, the question posed above is an attempt to mathematically characterize it. However, practically speaking, it is rather something that can not be described precisely but would be obvious to recognize if presented upon. From that point of view, the above question is just a reflection of our success/failure in finding a way to fit our feeling into an exact mathematical description. We do believe that over time one can develop a better formulation but until then we will rely on the one given above and on a bit of a subjective individual feeling. Along the same lines then, everything that we will write below should in a way be prefaced by such a statement.

%%%%%%%%%%%%%%%%%%%%%%%%%%%%%%%%%%%%%%%%%%%%%%%%%%%%%%%%%%%
\subsection{Restrictions} \label{sec:restrictions}
%%%%%%%%%%%%%%%%%%%%%%%%%%%%%%%%%%%%%%%%%%%%%%%%%%%%%%%%%%%

There are several comments that we believe are in place. They, in first place, refer to the restrictions we posed in the above question.

\begin{enumerate}

\item In the posed question we insist that the components of $A$ are i.i.d standard normal random variables. That may not necessarily be the right way to capture the universal capabilities of $\ell_1$ or for that matter the universal capabilities of any other algorithm. Still, it is our belief that such a statistical choice is the least harmful. In other words, if we assume that $A$ has a different type of randomness one then may ask why such a randomness is any more universal than say Gaussian. While we indeed restricted randomness of $A$ we believe that we did it in a fairly harmless way.

\item Another restriction that we introduced is the restriction on $\tilde{\x}$. This restriction may be a bit problematic if, for example, one works hard to select a particulary good/bad set of non-zero locations for a particularly good/bad matrix $A$. However, if $A$ is comprised of i.i.d. standard normals then this choice seems harmless as well. Of course, if a different $A$ is to be considered then restricting signs and locations can substantially bias $\tilde{\x}$. Also, although it is not necessary, we suggest one uniformly randomly select locations and signs and then fix them (given the rotational invariance of rows of $A$ this may sound as if unnecessary i.e. one can alternatively take any set of non-zero locations and any combination of signs). However, one eventually may want to upgrade Question 1 to include different matrices $A$ and then random choice of locations and signs of $\tilde{\x}$ may be needed.

\item Our choice of polynomial algorithms can also be problematic. For example, there are many algorithms that are provably polynomial but with running time that can hardly ever be executed practically. More importantly, by insisting that the algorithms are polynomial we are potentially excluding some of the random algorithms or those whose running time depends on the values of the input (which in our case are random!). This is probably one of the major issues with Question 1. It is possible that not much would change even if we allow, say, algorithms that are with overwhelming probability (taken over their own randomness or even over the randomness of the problem itself or even over both of them) polynomial. For example, if AMP was able to give a performance characterization higher than the one that BP gives, the answer to Question 1 would still not be yes. One would have to argue that AMP is a polynomial algorithm. That is exactly where the problems of polynomiality may appear. One could occasionally have problems arguing that typically super-fast random algorithms are in the worst-case polynomial. Moreover, one should as well be careful how the worst-case is interpreted, i.e. is it interpreted over problem instances or over its algorithm's own randomness. We do, however, believe that if the polynomiality is a stiff restriction one can relax it to polynomial with overwhelming probability, where, as mentioned above, randomness would be over both, the problem instances as well as potential random structure of the algorithm.

\end{enumerate}

%%%%%%%%%%%%%%%%%%%%%%%%%%%%%%%%%%%%%%%%%%%%%%%%%%%%%%%%%%%
\subsection{Redirecting a challenge} \label{sec:challenge}
%%%%%%%%%%%%%%%%%%%%%%%%%%%%%%%%%%%%%%%%%%%%%%%%%%%%%%%%%%%

If one can come to terms with deficiencies of the question that we posed then it may not be a bad idea to revisit the timeline of the problem it addresses. As is well known, under-determined linear systems with sparse solutions have been around for a long time. Consequently, a host of ways to attack them is known (in fact, we briefly discussed some of them in Section \ref{sec:back}). For a long time it had been a prevalent opinion that BP is a solid heuristic when it comes to increasing recoverable sparsity. Such a popular believe was analytically justified for the first time in seminal works \cite{CRT,DonohoPol,DonohoUnsigned}. Moreover, the results of \cite{DonohoPol,DonohoUnsigned} in a large dimensional and statistical context provided the exact performance characterization of BP. Initial success of \cite{CRT,DonohoPol,DonohoUnsigned} then generated enormous interest in sparse problems in many different fields. The set of achieved results does not seem exhaustable and as if growing on a daily basis. Impressive results have been achieved across a variety of disciplines and range from various algorithmic implementations to specific applications and needed adaptations.

Our own interest is on a purely mathematical level. From a purely mathematical point of view, Question 1 (with its all above mentioned deficiencies)  in our mind stands as a key test on the path of almost any improvement in recoverable sparsity characterizations. Providing answer yes to Question 1 is basically a guarantee that a mathematical improvement is possible. Now, looking back at what was done in last $10$ two lines of work that we mentioned in the previous sections are of particular interest. One is the line that follows the design and analysis of AMP and the other one is our own revisit of BP. However, not much progress seems to have been made as far as providing answer yes to Question 1 in any of these lines (and for that matter in any other line of work known to us). Namely, while both results, \cite{DonMalMon09,BayMon10} and \cite{StojnicCSetam09,StojnicUpper10} are incredible feat on their own, not only are they not moving the characterization obtained by Donoho in \cite{DonohoPol,DonohoUnsigned}, they are actually reestablishing it in a different way. Reestablishing Donoho's results is of course a fine mathematical achievement. However, when viewed through the prism of establishing answer yes to Question 1 reestablishing Donoho's results is a somewhat pessimistic progress.

More specifically, our own results, for example, in a way hint that the best one can do through a convex type of relaxation is probably what BP does. On the other hand situation may be even worse if one looks at AMP and results obtained in \cite{DonMalMon09,BayMon10}. It is almost unbelievable that a different algorithm (in this case the AMP) achieves exactly the same performance as BP. Since it does happen one naturally wonders how is it possible. One simple way would be that AMP essentially just solves BP, though in a very clever and efficient way (if this would turn out to be indeed true then, as far as moving up the curve in Figure \ref{fig:weak} is concerned, things may not be overly pessimistic). On the other hand, if AMP is indeed a fundamentally different approach then one may start thinking weather or not lifting the curve in Figure \ref{fig:weak} is actually possible within the frame of Question 1. And since there is currently really no evidence either way one simply wonders if it is already a time to start looking at Question 1 with the idea of providing answer no.

Since we have not looked at Question 1 from that perspective we can not really comment much as to what are the chances that the answer is indeed no. On the other had it has been almost $10$ years since Donoho created his results, almost $5$ years since we created our own, and probably as long since the results of \cite{DonMalMon09,BayMon10} were created. Given the massive interest that this field has seen in last $10$ years one would expect that if the answer to Question 1 is yes then it would have been already established. Of course, one can then alternatively argue the other way around as well. Namely, if the answer to Question 1 is no, wasn't there enough time during the last decade to establish it. We of course do not know if there was enough time for establishing any definite answer to Question 1. However, it is our belief that the majority of mathematical work was concentrated at establishing results that would imply answer yes to Question 1. If our belief is even remotely close to the truth then one can realistically ask if it is really a time to redirect the challenge and try to look at ways that would lead to providing answer no to Question 1.

As we have stated above, we do not know what the answer to Question 1 is. However, given that we expressed our belief that it is not impossible that the answer is actually no, it would be in fact reasonable that we provide at least some information as to which way we are leaning. Well, our position is somewhat funny but certainly worth sharing: we work believing that the answer is yes but if we were to bet we would bet that the answer is no. Of course this position is massively hedged but in our view seems reasonable. Namely, if it turns out that the answer is yes we would need to pay but would in return get to see the show which seems as a pretty nice option (if there is to be a show we firmly believe that it must be a big one!). On the other hand if there is no show we would get overreimbursed for the ticket we actually never had which is not that bad either. As far as our preference goes though, we would still prefer to see the show!

%%%%%%%%%%%%%%%%%%%%%%%%%%%%%%%%%%%%%%%%%%%%%%%%%%%%%%%%%%%%%%%%%%%%%%%%%%%%%%%%
\section{Further considerations}
\label{sec:furcons}
%%%%%%%%%%%%%%%%%%%%%%%%%%%%%%%%%%%%%%%%%%%%%%%%%%%%%%%%%%%%%%%%%%%%%%%%%%%%%%%%

%%%%%%%%%%%%%%%%%%%%%%%%%%%%%%%%%%%%%%%%%%%%%%%%%%%%%%%%%%%%%%%%%%%%%%%%%%%%%%%%
\subsection{What after Question 1}
\label{sec:aftquest1}
%%%%%%%%%%%%%%%%%%%%%%%%%%%%%%%%%%%%%%%%%%%%%%%%%%%%%%%%%%%%%%%%%%%%%%%%%%%%%%%%

In the previous section we discussed a possible shift in the approach to answering Question 1. A very important point to make is that even if one is able to answer Question 1 the whole story is not over. In this subsection we present what in our view would be further points of interest once Question 1 is settled.

If it turns out that the answer is no, then in a way the value of a majority of the work done in the previous decade would be even higher. As we have mentioned above a majority of the work done in last decade was related to polynomial algorithms (or those that are highly likely to be polynomial) and part of the $(\beta,\alpha)$ plane below the curve given in (\ref{eq:thmweaktheta2}) and Figure \ref{fig:weak}. In that sense the contribution of line of work initiated in \cite{DonMalMon09,BayMon10} would be pretty much invaluable.

On the other hand if it turns out that the answer to Question 1 is yes then naturally a variety of further questions will appear. The first next in our mind would be:

\textbf{Question 2:} Assuming that the answer to Question 1 is yes, can one then determine an alternative curve say $(\beta^{(opt)},\alpha^{(opt)})$ for which the answer to Question 1 is no? Along the same lines can it happen that there is no such a curve that is below a straight line at $1$?

Then one can go further and assuming that the answer to the first part of Question 2 is yes but the answer to the second part of Question 2 is no, ask the following:

\textbf{Question 3:} Assuming that the answer to the first part of Question 2 is yes, can one then lower curve $(\beta^{(opt)},\alpha^{(opt)})$ until the answer to Question 1 is no?

Settling all these questions would in our mind be a way to deepen our understanding of a polynomial solvability of under-determined linear systems with sparse solutions.

%%%%%%%%%%%%%%%%%%%%%%%%%%%%%%%%%%%%%%%%%%%%%%%%%%%%%%%%%%%%%%%%%%%%%%%%%%%%%%%%
\subsection{What after Questions 2 nd 3}
\label{sec:aftquest123}
%%%%%%%%%%%%%%%%%%%%%%%%%%%%%%%%%%%%%%%%%%%%%%%%%%%%%%%%%%%%%%%%%%%%%%%%%%%%%%%%

An important scenario that may play out when settling the above questions is that the ultimate curve $(\beta^{(opt)},\alpha^{(opt)})$ (under the premises of Question 1) is not the straight line at $1$ (for example, answer no to Question 1 immediately forces such a scenario). Such a scenario would be a great opportunity to revive studying random hardness within the current complexity theory framework. In our mind such a view of hardness portion of the traditional complexity theory is an important aspect both, practically and theoretically. Unfortunately, it seems a bit premature to start looking at it right now for a variety reasons. First, even in a general complexity theory there are fewer results that relate to random hardness then to typical notion of worst-case hardness/completeness. Second, we are not even sure that the current setup of random hardness/completeness has been well established/investigated even on way more popular optimization or decision problems.

Still, it is important to note that if one starts attacking Question 1 with an ambition to show that the answer is no, then the above mentioned random hardness concepts should probably be revisited and their meaning reunderstood and quite possibly even adapted to better fit the scope of the story presented here. The idea of this paper is just to hint that there may be a time to think about other directions when it comes to studying linear systems. We then consequently refrain from a further detailed discussion about this here, but mention that all these problems seem to be at a cutting edge of what we envision as a future prospect for studying under-determined systems with sparse solutions.

%%%%%%%%%%%%%%%%%%%%%%%%%%%%%%%%%%%%%%%%%%%%%%%%%%%%%%%%%%%%%%%%%%%%%%%%%%%%%%%%
\section{Conclusion}
\label{sec:conc}
%%%%%%%%%%%%%%%%%%%%%%%%%%%%%%%%%%%%%%%%%%%%%%%%%%%%%%%%%%%%%%%%%%%%%%%%%%%%%%%%

In this paper we revisited under-determined systems of linear equations with sparse solutions. We looked at a particular type of mathematical problems that arise when studying such systems. Namely, we looked at the characterizations of relations between the size of the system and the sparsity of the solutions so that the systems are solvable in polynomial time.

We started by giving a brief overview of the results that we considered as mathematically most important for a direction of study that we wanted to popularize. We then made several observations related to the pace of progress made in last $10$ years. When it comes to studying polynomial algorithms and their abilities to solve a class of random under-determined linear systems, our main observation is that there has been a somewhat limited progress as to what the ultimate performance characterization of such algorithms is. We then raised a question which in a way asks whether is it possible that the performance characterizations of two known algorithms (namely, BP and AMP) could in fact be the optimal ones when it comes to polynomial algorithms. We believe that this will stimulate a further discussion in this direction in a host of mathematical fields.

%\newpage1
%\setcounter{page}{1}
\begin{singlespace}
\bibliographystyle{plain}
\bibliography{ReDirChallRefs}
\end{singlespace}

\end{document}